\documentclass[11pt]{article}
\usepackage{amsthm, amsmath, amssymb, amsfonts, url, booktabs, tikz, setspace, fancyhdr, bm}
\usepackage{hyperref}
\usepackage{geometry}
\usepackage{hyperref, enumerate}
\usepackage[shortlabels]{enumitem}
\usepackage[babel]{microtype}
\usepackage[english]{babel}
\usepackage[capitalise]{cleveref}
\usepackage{comment}
\usepackage{bbm}
\usepackage{csquotes}
\usepackage{mathabx}
\usepackage{tikz}
\usepackage{graphicx}
\usepackage{float}
\usepackage{amsmath}

\counterwithin{figure}{section}

\newtheorem{theorem}{Theorem}[section]

\newtheorem{lemma}[theorem]{Lemma}
\newtheorem{cor}[theorem]{Corollary}

\newtheorem{claim}[theorem]{Claim}

\theoremstyle{definition}
\newtheorem{defn}[theorem]{Definition}
\newtheorem*{defn-non}{Definition}
\newtheorem{exam}[theorem]{Example}

\newtheorem{ques}[theorem]{Question}

\newlist{Case}{enumerate}{3}
\setlist[Case, 1]{%
    label           =   {\bfseries Case \arabic*.},
    labelindent=1em ,labelwidth=1cm, labelsep*=1em, leftmargin =!
}
\setlist[Case, 2]{%
    label           =   {\bfseries Subcase \arabic{Casei}.\arabic*.},
    labelindent=-1em ,labelwidth=1cm, labelsep*=1em, leftmargin =!
}
\setlist[Case, 3]{%
    label           =   {\bfseries Subsubcase \arabic{Casei}.\arabic{Caseii}.\arabic*.},
    labelindent=-1em ,labelwidth=1cm, labelsep*=1em, leftmargin =!
}

\newenvironment{poc}{\begin{proof}[Proof of claim]}{\end{proof}}

\newcommand{\ceil}[1]{\lceil #1\rceil}
\newcommand{\floor}[1]{\lfloor #1\rfloor}

\newcommand{\aaa}{\boldsymbol{a}}
\newcommand{\bb}{\boldsymbol{b}}
\newcommand{\cc}{\boldsymbol{c}}

\newcommand{\vv}{\boldsymbol{v}}

\newcommand{\uu}{\boldsymbol{u}}
\newcommand{\pp}{\boldsymbol{p}}
\newcommand{\qq}{\boldsymbol{q}}
\newcommand{\xx}{\boldsymbol{x}}
\newcommand{\yy}{\boldsymbol{y}}
\newcommand{\DD}{\boldsymbol{D}}

\usepackage{todonotes}

\newcommand{\im}{\mathrm{Im}}
\newcommand{\enc}{\mathrm{Enc}}
\newcommand{\wt}{\mathrm{wt}}
\newcommand{\vol}{\mathrm{Vol}}
\newcommand{\Z}{\mathbb{Z}}

\title{Optimal redundancy of function-correcting codes }

\author{
Gennian Ge\thanks{School of Mathematical Sciences, Capital Normal University, Beijing, China. Email: gnge@zju.edu.cn. Gennian Ge is supported by the National Key Research and Development Program of China under Grant 2020YFA0712100, the National Natural Science Foundation of China under Grant 12231014, and Beijing Scholars Program.}
\and 
Zixiang Xu\thanks{Extremal Combinatorics and Probability Group (ECOPRO), Institute for Basic Science (IBS), Daejeon, South Korea. Email: zixiangxu@ibs.re.kr. Supported by IBS-R029-C4.}
\and 
Xiande Zhang\thanks{School of Mathematical Sciences, University of Science and Technology of China,
Hefei, 230026, Anhui, China, and Hefei National Laboratory, University of Science and Technology of China, Hefei, 230088, Anhui, China. Emails:drzhangx@ustc.edu.cn, zyjshuxue@mail.ustc.edu.cn. Xiande Zhang and Yijun Zhang are supported by the National Key Research and Development Programs of China 2023YFA1010201 and 2020YFA0713100, the
NSFC under Grants No. 12171452 and No. 12231014, and the Innovation Program for Quantum Science and Technology (2021ZD0302902).}
\and
Yijun Zhang\footnotemark[3]
}

\begin{document}
\maketitle

\begin{abstract}
Function-correcting codes (FCCs), introduced by Lenz, Bitar, Wachter-Zeh, and Yaakobi, protect specific function values of a message rather than the entire message. A central challenge is determining the optimal redundancy—the minimum additional information required to recover function values amid errors. This redundancy depends on both the number of correctable errors \(t\) and the structure of message vectors yielding identical function values. While prior works established bounds, key questions remain, such as the optimal redundancy for functions like Hamming weight and Hamming weight distribution, along with efficient code constructions. In this paper, we make the following contributions: 
\begin{enumerate}  
    \item For the Hamming weight function, we improve the lower bound on optimal redundancy from \(\frac{10(t-1)}{3}\) to \(4t - \frac{4}{3}\sqrt{6t+2} + 2\). On the other hand, we provide a systematical approach to constructing explicit FCCs via a novel connection with Gray codes, which also improve the previous upper bound from \(\frac{4t-2}{1 - 2\sqrt{\ln(2t)/(2t)}}\) to \(4t - \log t\). Consequently, we almost determine the optimal redundancy for Hamming weight function. 

    \item The Hamming weight distribution function is defined by the value of Hamming weight divided by a given integer \(T\in \mathbb{N}\). Previous work established that the optimal redundancy is \(2t\) when \(T > 2t\), while the case \(T \leq 2t\) remained unclear. We show that the optimal redundancy remains \(2t\) when \(T \geq t+1\). However, in the surprising regime where \(T = o(t)\), we achieve near-optimal redundancy of \(4t - o(t)\). Our results reveal a significant distinction in behavior of redundancy for distinct choices of $T$.  
\end{enumerate}
\end{abstract}

\section{Introduction}
\subsection{Background}
In conventional communication systems, a sender transmits a message to a receiver over a noisy channel, typically assuming that every part of the message is equally important. The primary objective is to design an \emph{error-correcting code} (ECC) and an appropriate decoder to recover the entire message with high fidelity. However, in many practical scenarios~\cite{Ahlswede1981information,boyarinov1981linear,masnick1967linear,Orlitsky2001coding}, only specific attributes or functions of the message are of primary interest. While recovering the full message enables the evaluation of these functions, doing so can be inefficient when the message is large and the output of function is small.  

To address this inefficiency, Lenz, Bitar, Wachter-Zeh, and Yaakobi~\cite{lenz2023function} introduced the concept of \emph{function-correcting codes} (FCCs), which focus on protecting specific functions of messages rather than the entire message itself. When only a particular attribute needs to be preserved, redundancy can be significantly reduced compared to traditional ECCs. Moreover, in applications where maintaining the original form of data is desirable, such as distributed computing~\cite{Kuzuoka2017distributed,Orlitsky2001coding,Wei2023robust} and distributed storage~\cite{Shutty2022bandwidth}, the authors in~\cite{lenz2023function} proposed a systematic encoding approach for FCCs. In this framework, redundancy is appended to the original message, ensuring efficient function recovery while keeping the message structure intact, as illustrated in~\cref{fig:intro_fcc_problem}.

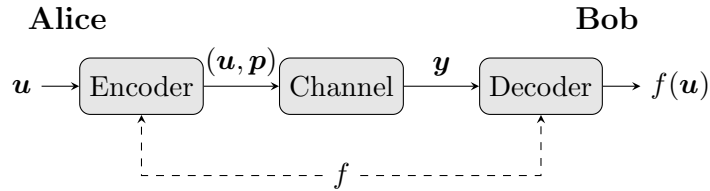
\begin{figure}[t]
	\centering
	\begin{tikzpicture}[>=stealth]
		\node (u) {$\uu$};
		\node[right = 0.5cm of u, rectangle, rounded corners, fill=gray!20!, draw, minimum height=0.75cm] (enc) {Encoder};
		\node[above left = 0.25cm and -0.5cm of enc] (alice) {\bfseries\large{Alice}};

		\node[draw,rectangle, rounded corners, fill=gray!20!, right = 1 cm of enc, minimum height=0.75cm] (channel) {Channel};
		\node[draw,rectangle, rounded corners, fill=gray!20!, right = 1 cm of channel, minimum height=0.75cm] (decoder) {Decoder};
		
		\node[below = 0.5cm of channel] (fu) {$f$};
		
		\node[right = 0.5cm of decoder](fy) {$f(\uu)$};
		\node[above right= 0.25cm and -0.5cm of decoder] (bob) {\bfseries\large{Bob}};
		
		\draw[->] (u) -- (enc);%
		\draw[->] (enc) -- node [above] {$(\uu,\pp)$} (channel);
		\draw[->] (channel) -- node [above] {$\yy$} (decoder);
		\draw[->] (decoder) -- (fy);
		\draw[->, dashed] (fu) -| (enc);
		\draw[->, dashed] (fu) -| (decoder);
	\end{tikzpicture}

	\caption{Overview of the systematic FCC framework: In this setup, Alice wishes to transmit a message \(\uu\) (\(\uu\) can be viewed as a binary vector) to Bob, where a specific attribute \( f(\uu) \) is of primary interest. To ensure Bob can accurately recover this attribute despite potential transmission errors, Alice encodes \(\uu\) into a structured codeword \(\cc = (\uu, \pp)\), where \(\pp\) provides redundancy. Upon receiving a possibly corrupted version \(\yy\) of \(\cc\), Bob leverages his knowledge of \( f \) to reliably infer \( f(\uu) \), even without full error correction.}
	\label{fig:intro_fcc_problem}
\end{figure}

Recognizing that the primary advantage of FCCs lies in their ability to reduce redundancy, the authors of \cite{lenz2023function} aimed to determine the minimal redundancy necessary for accurately recovering key attributes. By focusing on essential information, FCCs provide an efficient alternative to traditional error-correcting methods in communication and storage systems, and their objective was to identify the optimal redundancy for FCCs designed for specific functions. By establishing connections between FCCs and irregular-distance codes, the authors of \cite{lenz2023function} derived both lower and upper bounds on the optimal redundancy of FCCs for general functions, which will be introduced in~\cref{sec:Equivalent} and~\cref{sec:KnownResults}. These results were further applied to specific cases, including locally binary functions, Hamming weight functions, Hamming weight distribution functions, Min-Max functions, and certain real-valued functions. Building on this foundation, Premlal and Rajan in \cite{premlal2024function} extended the study of FCCs to linear functions, an important variant in this research direction. 

Exploring FCCs over different channels is also an intriguing area of study. For example, Xia, Liu and Chen~\cite{xia2024function} extended the concept of FCCs to symbol-pair read channels and obtained various new theoretical results. Such investigations deepen the understanding of FCCs and open up new possibilities for their applications. 


Formally, let \(\uu \in \Z_2^k\) represent a binary message, and consider a function \( f: \Z_2^k \to \im(f) \), where \(\im(f) := \{ f(\uu) : \uu \in \Z_2^k \}\) denotes the image of \(f\). The message is encoded using an encoding function \(\enc: \Z_2^k \to \Z_2^{k+r}\), defined as \(\enc(\uu) = (\uu, \pp(\uu))\). To emphasize the distinction, \(\uu \in \Z_2^k\) is referred to as the \emph{message vector}, \(\pp(\uu) \in \Z_2^r\) is called the \emph{redundancy vector}, and \(r\) is defined as the \emph{redundancy}. For positive integer \(n\) and two vectors $\uu,\vv\in\Z_2^n$, the Hamming distance between \(\uu\) and \(\vv\), denoted by \(d(\uu, \vv)\), is defined as the number of coordinates in which they differ.

The formal definition of function-correcting codes is as follows.

\begin{defn}
   Let \(k\), \(r\), and \(t\) be positive integers. An encoding function \(\enc: \Z_2^k \to \Z_2^{k+r}\), defined as \(\enc(\uu) = (\uu, \pp(\uu))\) for \(\uu \in \Z_2^k\), is said to define a \emph{function-correcting code} (FCC for short) for a function \(f: \Z_2^k \to \im(f)\) if, for all \(\uu_1, \uu_2 \in \Z_2^k\) satisfying \(f(\uu_1) \neq f(\uu_2)\), the following condition holds:
   \[
   d(\enc(\uu_1), \enc(\uu_2)) \geq 2t + 1.
   \]
   Here, \(d(\cdot, \cdot)\) denotes the Hamming distance. The \emph{optimal redundancy} \(r_{f}(k, t)\) is defined as the smallest integer \(r\) for which there exists a function-correcting code with an encoding function \(\enc: \Z_2^k \to \Z_2^{k+r}\) for the function \(f\).
\end{defn}

By this definition, for any vector \(\yy\) obtained by introducing at most \(t\) errors into \(\enc(\uu)\), the receiver can uniquely recover \(f(\uu)\), provided it has knowledge of the function \(f(\cdot)\) and the encoding function \(\enc(\cdot)\). The central problem in this area is to study the behavior of the optimal redundancy, denoted by \(r_f(k, t)\).

\begin{ques}
    Let \(t\) and \(k\) be positive integers and \(f: \Z_2^k \to \textup{Im}(f)\) be a function. Determine \(r_{f}(k, t)\).
\end{ques}

\textbf{Notations.}
For simplicity, we adopt the following notation:
For any integer \( M \geq 1 \), define \([M] = \{1, 2, \dots, M\}.\) For integers \( a \leq b \), define \([a, b] = \{a, a+1, \dots, b-1, b\}\). For an integer \( M \geq 1 \) and an integer \( s \), define \( s \bmod M \) as the unique representative element in the set \([0, M-1]\). For any integer \( 0 \leq i \leq k \), define the vector \((0^{k-i} 1^i) \in \Z_2^k\) as the vector whose first \(k-i\) coordinates are 0 and the remaining \(i\) coordinates are 1. Throughout this paper, \(\log{x}\) denotes \(\log_{2}{x}\).


\subsection{Equivalent reformulation}\label{sec:Equivalent}

One of the key connections established in~\cite{lenz2023function} is the relationship between function-correcting codes and irregular-distance codes. For a matrix \(\DD\), we denote by \([\DD]_{ij}\) the \((i,j)\)-th entry of \(\DD\). To formalize this connection, the authors in~\cite{lenz2023function} introduced the distance requirement matrix of a function \(f\) as follows.

\begin{defn}
    For an integer \(M \geq 1\), let \(\uu_1, \uu_2, \dots, \uu_M \in \Z_2^k\). The distance requirement matrix \(\DD_f(t, \uu_1, \dots, \uu_M)\) of a function \(f\) is an \(M \times M\) matrix with entries defined as:
    \[
    [\DD_f(t, \uu_1, \dots, \uu_M)]_{ij} =
    \begin{cases}
        \max\{2t + 1 - d(\uu_i, \uu_j),0\}, & \text{if } f(\uu_i) \neq f(\uu_j), \\
        0, & \text{otherwise},
    \end{cases}
    \]
    where \(d(\cdot, \cdot)\) denotes the Hamming distance.
\end{defn}

Let \(\mathcal{P} = \{\pp_1, \pp_2, \dots, \pp_M\} \subseteq \Z_2^r\) be a code of length \(r\) and cardinality \(M\). Here, \(r\) is chosen as length of the code, as it can later be related to the redundancy of FCCs. Irregular-distance codes are formally defined as follows.

\begin{defn}
    Let \(\DD \in \mathbb{N}_0^{M \times M}\). A collection of vectors \(\mathcal{P} = \{\pp_1, \pp_2, \dots, \pp_M\}\) is called a \(\DD\)-code if there exists an ordering of its codewords such that \(d(\pp_i, \pp_j) \geq [\DD]_{ij}\) for all \(i, j \in [M]\). 

    Furthermore, we define \(N(\DD)\) as the smallest integer \(r\) for which there exists a \(\DD\)-code of length \(r\). If \([\DD]_{ij} = D\) for all \(i \neq j\), we denote the corresponding \(N(\DD)\) as \(N(M, D)\).
\end{defn}

The following result in~\cite{lenz2023function} shows that above ways to define redundancy are equivalent.
\begin{theorem}[\cite{lenz2023function}]\label{thm:equivlent}
    For any function $f: \mathbb{Z}_2^k \to \im(f)$,
	 \[
  r_f(k, t) =  N(\DD_f(t, \uu_1, \ldots, \uu_{2^k})).
  \]
    
\end{theorem}
An alternative and highly effective approach to studying this problem is that, instead of directly investigating the properties of original $\DD$-codes, it can sometimes be more advantageous to focus on certain local regions of the codes. 

\begin{defn}
    For any $\mathcal{M}:=\{\uu_1,\uu_2,\dots , \uu_m\}\subseteq \Z_2^k$, we define $r_f(k,t,\mathcal{M})$ as the smallest $r$ such that there exists a function-correcting code with encoding function $\enc:\mathcal{M} \to \mathcal{M}\times \Z_2^{r}$ for the function $f$, i.e., for any $\uu_i,\uu_j\in \mathcal{M}$ with $f(\uu_i)\neq f(\uu_j)$, we have $d(\enc(\uu_i),\enc(\uu_j))\ge 2t+1$.
\end{defn}

It is obvious that $N(\DD_f(t,\uu_1,\dots,\uu_{2^k}))\geq N(\DD_f(t,\uu_1,\dots,\uu_{M}))$ for any $M\leq 2^k$, thus the following corollary holds.
\begin{cor}[\cite{lenz2023function}]\label{cor1}
	Let $M,k,t$ be some positive integers, and let $\uu_1,\dots,\uu_{M} \in \Z_2^k$ be arbitrary different vectors. Then, the redundancy of a function-correcting code is at least
	$$r_f(k,t) \geq N(\DD_f(t,\uu_1,\dots,\uu_{M})).$$
In particular, for any function $f$ with $|\im(f)|\geq2$,
	$$r_f(k,t)\geq 2t.$$
\end{cor}

\subsection{Known results}\label{sec:KnownResults}

The well-known Plotkin bound~\cite{Plotkin1960binary} and Gilbert-Varshamov bound~\cite{gilbert1952comparison,varshamov1957estimate} serve as fundamental tools in the study of coding theory. In~\cite{lenz2023function}, it is demonstrated that these two bounds can be extended to irregular-distance codes.

\begin{lemma}[\cite{lenz2023function}]\label{Plotkin Bound}
    Let \(M\) be a positive integer. For any \(M \times M\) distance requirement matrix \(\DD \in \mathbb{N}^{M \times M}\), we have
    \[
    N(\DD) \geq 
    \begin{cases} 
    \frac{4}{M^2} \sum\limits_{\substack{i,j \in [M], i < j}} [\DD]_{ij}, & \text{if } M \text{ is even}, \\
    \frac{4}{M^2 - 1} \sum\limits_{\substack{i,j \in [M], i < j}} [\DD]_{ij}, & \text{if } M \text{ is odd}.
    \end{cases}
    \]
\end{lemma}

To this end, we define \(\vol(r, d) = \sum_{i=0}^{d} \binom{r}{i}\) as the size of the binary radius-\(d\) Hamming sphere over vectors of length \(r\).

\begin{lemma}[\cite{lenz2023function}]\label{gv bound}
    Let \(M\) be a positive integer. For any \(M \times M\) distance requirement matrix \(\DD \in \mathbb{N}^{M \times M}\) and any permutation \(\pi: [M] \to [M]\),
    \[
    N(\DD) \leq \min_{r \in \mathbb{N}} \left\{ r : 2^{r} > \max_{j \in [M]} \sum_{i=1}^{j-1} \vol(r, [\DD]_{\pi(i)\pi(j)} - 1) \right\}.
    \]
\end{lemma}

For codes of small cardinality, i.e., when the code size is of the same order of magnitude as the minimum distance, the above bound can be improved based on Hadamard codes.

\begin{lemma}[\cite{lenz2023function,xia2024function}]\label{GV hadamard}
    For any \(M, D \in \mathbb{N}\) with \(D \geq 10\) and \(M \leq D^2\),
    \[
    N(M, D) \leq \frac{2D - 2}{1 - 2\sqrt{\ln(D)/D}}.
    \]
\end{lemma}

In specific applications, certain functions have garnered significant attention due to their practical relevance and theoretical importance. A number of such functions have been studied in the context of FCCs, with known results for their optimal redundancy. Below, we summarize the key results for specific types of functions in~\cref{tab:compara}.

\begin{table}[htbp]
    \centering
    \renewcommand{\arraystretch}{1.5} 
    \begin{tabular}{|c|c|c|c|}
        \hline
        Function & Parameter & Lower bound & Upper bound \\ \hline
        Binary & - & $2t$ & $2t$ \\ \hline
        Hamming weight wt$(\uu)$ & - & $\displaystyle\frac{10(t-1)}{3}$ & \(\displaystyle\frac{4t-2}{1 - 2\sqrt{\ln(2t)/(2t)}}\) \\ \hline
        Locally binary & $E$ & $2t$ & $2t$ \\ \hline
        Hamming weight distribution $\Delta_{T}(\uu)$ & $T \geq 2t+1$ & $2t$ & $2t$ \\ \hline
        Min-max $\textup{mm}_{w}(\boldsymbol{u})$ & $w \gg 2t$ & $2\log{w}+(t-2)\log{\log{w}}$ & $2\log{w}+t\log{\log{w}}$ \\ \hline
    \end{tabular}
    \caption{Some theoretical bounds for certain functions in~\cite{lenz2023function}}
    \label{tab:compara}
\end{table}

\section{Our contributions}
Existing results on the optimal redundancy of FCCs~\cite{lenz2023function} establish a robust theoretical framework, yet significant potential for improvement remains. Our primary objective is to narrow the gap between these bounds and to present systematic, explicit constructions of FCCs for a variety of functions. We then introduce our main results separately.

\subsection{Hamming weight function}
The \emph{Hamming weight function}, denoted as \( f(\uu) = \text{wt}(\uu) \), plays a crucial role in the study of FCCs due to its simplicity and widespread relevance in information theory, coding theory, and practical applications such as error correction and data compression. The Hamming weight function counts the number of non-zero bits in a binary vector \(\uu\), which is an essential operation in various coding schemes. More precisely, we define the \emph{Hamming weight} of a vector \(\uu\), denoted by \(\wt(\uu)\), as the Hamming distance between \(\uu\) and the zero vector. 

In many practical systems, reducing redundancy is highly desirable due to bandwidth and storage constraints. FCCs that aim to protect the Hamming weight enable more efficient use of redundancy. In this case, the redundancy should be just enough to allow the recovery of the Hamming weight under noisy conditions, without needing to recover the entire message. This provides a targeted, function-specific protection strategy that is often more efficient than traditional error-correcting codes (ECCs) that attempt to recover the entire message. 

From now on, for \( 0\le i,j\le k \), define \( \uu_i=(0^{k-i}1^{i}) \). In~\cite{lenz2023function}, the following lemma demonstrates that it is sufficient to specify the distance demands between the vectors \( \{\uu_i\}_{i=0}^{k} \).

\begin{lemma}[\cite{lenz2023function}]\label{lem:weight functionalcoding}
    For any positive integers \(k\) and \(t\), we have 
    \[
    r_{\wt}(k,t)=N(\DD_{\wt}(t,\uu_0,\uu_1,\dots,\uu_k)).
    \]
\end{lemma}

As established in~\cite{lenz2023function}, the optimal redundancy for the Hamming weight function \(r_{\wt}(k,t)\) has known bounds, with the lower bound derived using Plotkin-like bounds for irregular-distance codes, while the upper bound derived from the Gilbert-Varshamov bound.

\begin{theorem}[\cite{lenz2023function}]\label{original bound}
    For any integer $k>2$, $r_{\wt}(k,1)=3$ and $r_{\wt}(k,2)=6$. Further, for $t \geq 5$ and $k>t$,
	$$ \frac{10(t-1)}{3} \leq r_{\wt}(k,t) \leq \frac{4t-2}{1-2\sqrt{\ln(2t)/(2t)}}. $$
\end{theorem}

Our main contribution is that, we nearly determine the optimal redundancy for the Hamming weight function. 

\begin{theorem}\label{lower bound weight function}
   For any integers $t\ge 5$ and $k>t$, we have
    \begin{equation*}
        r_{\wt}(k,t)\geq 4t-\frac{4}{3}\sqrt{6t+2}+2.
    \end{equation*}
\end{theorem}
Moreover, we tighten the previously best-known upper bound from \((4+o(1))t\) to \((4-o(1))t\) as follows. Notably, the bound in~\cref{thm:constructionGary} is obtained through explicit constructions. Rather than relying on the existential result from~\cite{lenz2023function}, we reveal a novel connection between FCCs and the so-called Gray codes, which enables us to construct explicit FCCs with reduced redundancy.
\begin{theorem}\label{thm:constructionGary}
    Let $k,t\in\mathbb{N}$ be positive integers. When $2^{\ceil{\log(2t+1)}}-(2t+1)\le 2^{\frac{2}{3}\ceil{\log(2t+1)}}$, there exists a function-correcting code with encoding function $\enc:\mathbb{Z}_{2}^{k}\rightarrow \mathbb{Z}_{2}^{k+r}$ for $r=4t + p - \lceil \log(2t + 1) \rceil$,
where \(p\) can be interpreted as the number of ones in the binary representation of \(2^{\lceil \log(2t + 1) \rceil} - (2t+1)\). Moreover, when $t+1$ is a power of two, there exists an explicit construction of function-correcting code with redundancy $4t-\ceil{\log{t}}$.
\end{theorem}

\subsection{Hamming weight distribution}

Let \(T\) be an arbitrary positive integer. The Hamming weight distribution function, an extension of the Hamming weight function, is defined as:
\[
f(\uu) = \Delta_T(\uu) = \left\lfloor \frac{\text{wt}(\uu)}{T} \right\rfloor,
\]
where \(\text{wt}(\uu)\) denotes the Hamming weight of the binary vector \(\uu\).

One can see that \(f(\uu)\) maps each input \(\uu\) to a value based on its Hamming weight, divided by \(T\). This creates a partition of the set of all binary vectors into bins, where each bin corresponds to a specific range of Hamming weights.

It was pointed out in~\cite{lenz2023function} that when $T\ge 4t+1$, $\Delta_{T}(\uu)$ is indeed $2t$-locally binary, therefore one can obtain that $r_{\Delta_T}(k, t) = 2t$ when $t\ge 4T+1$. Furthermore, in~\cite{lenz2023function} the authors completely determined the optimal redundancy when $T\ge 2t+1$.

\begin{theorem}[\cite{lenz2023function}]\label{thm:KnownDistribution}
    Let \( k, t, T \in \mathbb{N} \) be such that \( T \) divides \( k+1 \) and \( T \geq 2t+1 \). Then, the optimal redundancy for the Hamming weight distribution function satisfies  
    \[
    r_{\Delta_{T}}(k,t) = 2t.
    \]
\end{theorem}

This result indicates that for large enough \(T\), the redundancy required to recover the function evaluation is directly proportional to the error tolerance \(t\), and specifically, it is \(2t\). However, for smaller values of \(T\), particularly when \(T \leq 2t\), the theoretical understanding of redundancy has been limited, and the existing bounds are less effective or non-existent.

This work aims to study the case when \(T \leq 2t\), addressing the situation where the weight intervals are closely packed and potentially overlap. For these cases, the existing approach does not provide tight or practical bounds, leaving a gap in the theoretical framework of FCCs. Our contribution in this regard is to derive refined bounds that accurately capture the optimal redundancy for the Hamming weight distribution function when \(T \leq 2t\), thereby extending the applicability of FCCs in more constrained settings.

\begin{theorem}\label{thm:NewDistribution}
   Let $k,t,T$ be positive integers with $k> \sqrt{T(6t-T+3)}$, then 
    \[ r_{\Delta_{T}}(k,t)\ge 4t-\frac{4}{3}\cdot\sqrt{T(6t-T+3)}+2. \] 
\end{theorem}
Note that the Hamming weight function \(f(\uu)=\wt(\uu)\) can be viewed as the special case of Hamming weight distribution function \(f(\uu)=\floor{\frac{\wt(\uu)}{T}}\) when \(T=1\). Therefore~\cref{lower bound weight function} is a special case of~\cref{thm:NewDistribution}.

We also obtain some new upper bound in this case based on explicit constructions. In particular, we obtain the exact value of redundancy when $T\ge t+1$, which extends~\cref{thm:KnownDistribution}. Moreover, when $T=o(t)$, the upper bound $4t-o(t)$ is almost tight if $t$ is large enough.

\begin{theorem}\label{thm:constructionweightdistribution}
Let \( k, t \in \mathbb{N} \) be positive integers, and let \( \Delta_{T} \) be the Hamming weight distribution function.
\begin{enumerate} 
\item[\textup{(1)}] If \( T \geq t+1 \), there exists a function-correcting code with redundancy \( 2t \), which is optimal. \item[\textup{(2)}] If \( T \leq t \), define \( z := \ceil{\frac{2t+1}{T}} \). If \( 2^{\ceil{\log z}} - z \leq 2^{\frac{2}{3} \ceil{\log z}} \), then there exists a function-correcting code with redundancy \[ 4t+(p-1)T-\ceil{\log z} \cdot T+1, \] where \( p \) denotes the number of ones in the binary representation of \( 2^{\lceil \log z \rceil} - z \). In particular:
\begin{itemize}
\item If \( \frac{2t+1}{T} = 2^m - 1 \) for some integer \( m \), then the redundancy is \[ 4t - \ceil{\log \frac{2t+1}{T}} \cdot T + 1. \]
\item If \( t \) is sufficiently large, \( T = o(t) \), and \( 2^{\ceil{\log z}} - z \leq 2^{\frac{2}{3} \ceil{\log z}} \), then the redundancy is \(4t-o(t).\) \end{itemize}
\end{enumerate} 
\end{theorem}

\medskip
{\bf \noindent Structure of this paper.} 
The remainder of this paper is structured as follows. In~\cref{sec:HammingWeight}, we focus on proving~\cref{thm:NewDistribution}, which also implies~\cref{lower bound weight function}. The results concerning explicit constructions will be presented in~\cref{sec:construction upperbound}. Finally, in~\cref{sec:Concluding}, we conclude the paper and present new results on the symbol-pair channel, along with potential directions for further research.

\section{Lower bounds on redundancy: Proof of~\cref{thm:NewDistribution}}\label{sec:HammingWeight}

For \(0 \leq i, j \leq k\), recall that \(\uu_i = (0^{k-i}1^{i})\). For simplicity, let \(d_{ij}\) denote the distance \(d(\uu_i, \uu_j) = |i - j|\). Let \(T \in \mathbb{N}\). Without loss of generality, we assume that \(T\) divides \(k+1\). In this section, we primarily focus on FCCs for the weight distribution function \(f(\uu) = \Delta_T(\uu) = \floor{\frac{\textup{wt}(\uu)}{T}}\), where \(T \leq 2t\). In previous work~\cite{lenz2023function}, the authors demonstrated that if \(T \geq 2t + 1\), then \(r_{\Delta_T}(k, t)\) can be determined explicitly as \(2t\). However, the case when \(T \leq 2t\) remains unresolved. Surprisingly, we show that the behavior in this regime may differ significantly.

Before the formal proof, we first introduce the following auxiliary lemma, which extends~\cref{lem:weight functionalcoding} to the Hamming weight distribution function \(\Delta_T\).

\begin{lemma}\label{lem:weightdistribution functionalcoding}
    For any positive integers \(k\) and \(t\) and \(\{\uu_i = (0^{k-i}1^{i})\}_{i=0}^{k}\), we have 
    \[
    r_{\Delta_T}(k, t) = N(\DD_{\Delta_T}(t, \uu_0, \uu_1, \dots, \uu_k)).
    \]
\end{lemma}

\begin{proof}
    On one hand, by~\cref{cor1}, we have \( r_{\Delta_T}(k, t) \geq N(\DD_{\Delta_T}(t, \uu_0, \uu_1, \dots, \uu_k)) \). It then remains to prove \( r_{\Delta_T}(k, t) \le N(\DD_{\Delta_T}(t, \uu_0, \uu_1, \dots, \uu_k)). \)
   
    Let \(\pp_0, \pp_1, \ldots, \pp_k \in \mathbb{Z}_2^r\) be a sequence of redundancy vectors corresponding to \(\uu_0, \uu_1, \ldots, \uu_k\), respectively. By definition, for any \(0 \leq i \neq j \leq k\), 
    \[
    d((\uu_i, \pp_i), (\uu_j, \pp_j)) \geq 2t + 1.
    \]
   Let \(\uu, \vv \in \Z_2^k\) be distinct vectors with \(\wt(\uu) = i\) and \(\wt(\vv) = j\), where \(\floor{\frac{i}{T}} \neq \floor{\frac{j}{T}}\). First obviously we have \(\Delta_T(\uu) \neq \Delta_T(\vv)\). Moreover, we can see
    \[
    \begin{aligned}
        d((\uu, \pp_i), (\vv, \pp_j)) &= d(\uu, \vv) + d(\pp_i, \pp_j) \\
        &\geq |\wt(\uu) - \wt(\vv)| + d((\uu_i, \pp_i), (\uu_j, \pp_j)) - d(\uu_i, \uu_j) \\
        &\geq |i-j| + 2t + 1 - d_{ij} \\
        &= 2t + 1.
    \end{aligned}
    \]
    Then we can define an encoding function as \(\enc(\uu) = (\uu, \pp_{\wt(\uu)})\) for any message vector \(\uu \in \Z_2^k\), which implies 
    \[
    r_{\Delta_T}(k, t) \leq N(\DD_{\Delta_T}(t, \uu_0, \uu_1, \dots, \uu_k)).
    \]
 This completes the proof.
\end{proof}
 We then provide the formal proof of~\cref{thm:NewDistribution}.
\begin{proof}[Proof of~\cref{thm:NewDistribution}]

Let \(\pp_0, \pp_1, \ldots, \pp_k \in \mathbb{Z}_2^r\) be a sequence of redundancy vectors corresponding to \(\uu_0, \uu_1, \ldots, \uu_k\), respectively. That is to say, \(d((\uu_i,\pp_i),(\uu_j,\pp_j))\ge 2t+1\) for any \(0\le i\ne j\le k\). Without loss of generality, we may assume that the length of the sequence \(\{\pp_i\}_{i=0}^{k}\) is minimal among all possible sets of redundancy vectors satisfying this condition. Then, we have \(r_{\Delta_{T}}(k,t) = r\) according to Lemma~\ref{lem:weightdistribution functionalcoding}. For simplicity, in the proof of the following result, we assume that \(\sqrt{\frac{6t - T + 3}{T}}\) is an integer, as this assumption does not significantly affect the lower bound on \(r\). Let \(m := \sqrt{\frac{6t - T + 3}{T}}\). 

We aim to carefully select a subset \(A \subseteq \left[0, \frac{k+1}{T} - 1\right]\) of size \(m\) and define \(B = \bigcup_{i \in A} [iT, iT + T - 1] \subseteq [0, k]\), and then apply double counting to the total Hamming distances \( \{d(\pp_i,\pp_j)\}_{i,j\in B} \) with \(\floor{\frac{i}{T}}\ne \floor{\frac{j}{T}}\) to derive a lower bound for \(r\) by Lemma~\ref{lem:weightdistribution functionalcoding}. Note that \(A\) uniquely determines \(B\). After selecting \(A\), observe that for any distinct \(i, j \in B\) with \(\lfloor \frac{i}{T} \rfloor \neq \lfloor \frac{j}{T} \rfloor\), the following inequality holds:
\[
d((\uu_i, \pp_i), (\uu_j, \pp_j)) = d_{ij} + d(\pp_i, \pp_j) \geq 2t + 1.
\]
Consequently, summing over all ordered distinct pairs of \(i,j\in A\) and \(\Delta_{T}(\uu_i)\ne \Delta_{T}(\uu_j)\), we have
\[
\sum_{\substack{i,j \in B \\ \left\lfloor \frac{i}{T} \right\rfloor \neq \left\lfloor \frac{j}{T} \right\rfloor}} d((\uu_i, \pp_i), (\uu_j, \pp_j)) = \sum_{\substack{i,j \in B \\ \left\lfloor \frac{i}{T} \right\rfloor \neq \left\lfloor \frac{j}{T} \right\rfloor}} (d_{ij} + d(\pp_i, \pp_j)) \geq \sum_{\substack{i,j \in B \\ \left\lfloor \frac{i}{T} \right\rfloor \neq \left\lfloor \frac{j}{T} \right\rfloor}} (2t + 1) \geq 2(2t + 1) \binom{m}{2} T^2.
\]

On the other hand, for each \(h \in [r]\), let \(a_h\) denote the number of ones appearing in the \(h\)-th coordinate of \(\{\pp_a\}_{a \in B}\). Note that \(|B| = mT\). Applying double counting, we obtain
\[
\sum_{\substack{i,j \in B \\ \left\lfloor \frac{i}{T} \right\rfloor \neq \left\lfloor \frac{j}{T} \right\rfloor}} d(\pp_i, \pp_j) \leq \sum_{i \neq j \in B} d(\pp_i, \pp_j) = 2 \sum_{h=1}^r a_h (mT - a_h) \leq \frac{r m^2 T^2}{2}.
\]
Combining the above inequalities, we conclude that
\begin{equation}\label{weight2}
    \sum_{\substack{i,j \in B \\ \left\lfloor \frac{i}{T} \right\rfloor \neq \left\lfloor \frac{j}{T} \right\rfloor}} d_{ij} + \frac{r m^2 T^2}{2} \geq (2t + 1) m(m - 1) T^2.
\end{equation}

By definition, \( d_{ij} \) is independent of \( r \). Define  
\[
S := \sum_{\substack{i,j\in B \\ \floor{\frac{i}{T}} \neq \floor{\frac{j}{T}}}} d_{ij} = \sum_{\substack{i,j\in B \\ \floor{\frac{i}{T}} \neq \floor{\frac{j}{T}}}} |i - j|.
\] 

Therefore, to derive a better lower bound for \(r\), we aim to minimize \( S \). This requires carefully selecting the subset \(A\). The following claim is crucial, as it shows that it is sufficient to choose \(A\) as a sequence of consecutive integers.
\begin{claim}\label{claim:Consecutive2}
    $S$ is minimized when $A$ consists of $m$ consecutive non-negative integers.
\end{claim}
\begin{poc}
  
Suppose that some selection of \( A \) minimizes the value of \( S \), but \( A \) is not a sequence of consecutive integers. Let \( a \in A \) be the smallest element such that \( a \in A \) but \( a + 1 \notin A \). Now, define  
\[
A_{1} := \{ i : i \leq a, i \in A\} \quad \text{and} \quad A_{2} := \{i - 1 : i > a + 1, i \in A\}.
\]  
Furthermore, we set
\[B_{1}:=\bigcup_{i\in A_{1}}[iT,iT+T-1]\ \text{and}\ B_{2}:=\bigcup_{i\in A_{2}}[iT,iT+T-1].\] 
Then we have  
\[
\begin{aligned}
\sum_{\substack{i,j\in B \\ \floor{\frac{i}{T}} \neq \floor{\frac{j}{T}}}} d_{ij} &= \sum_{\substack{i,j\in B_1 \\ \floor{\frac{i}{T}} \neq \floor{\frac{j}{T}}}} d_{ij} + \sum_{\substack{i,j\in B_2 \\ \floor{\frac{i}{T}} \neq \floor{\frac{j}{T}}}} d_{ij} + 2\sum_{i \in B_{1}, j \in B_{2}} d_{ij} \\
&= \sum_{\substack{i,j\in B_{1} \\ \floor{\frac{i}{T}} \neq \floor{\frac{j}{T}}}} d_{ij} + \sum_{\substack{i,j\in B_2 \\ \floor{\frac{i}{T}} \neq \floor{\frac{j}{T}}}} d_{ij} + 2\sum_{i \in B_{1}, j \in B \setminus B_{1}} (d_{ij} - T) \\
&= \sum_{\substack{i,j\in B \\ \floor{\frac{i}{T}} \neq \floor{\frac{j}{T}}}} d_{ij} - 2T|B_{1}| \cdot |B_{2}|,
\end{aligned}
\]  
which contradicts the rule of selection of \( A \), as both \( B_{1} \) and \( B_{2} \) are non-empty. This completes the proof of claim.
\end{poc}

By the definition of \( d_{ij} = |i - j| \), we can assume without loss of generality that \( A = [0, m-1] \) and consequently \( B = [0, mT - 1] \). Then,  
\[
\begin{aligned}
    S &= \sum_{i \neq j \in B} |i - j| - \sum_{\substack{i \neq j \in B \\ \floor{\frac{i}{T}} = \floor{\frac{j}{T}}}} |i - j| \\
      &= \sum_{i \neq j \in B} |i - j| - \sum_{s=0}^{m-1} \sum_{\substack{i \neq j \in B \\ \floor{\frac{i}{T}} = \floor{\frac{j}{T}} = s}} |i - j| \\
      &= \sum_{i \neq j \in B} |i - j| - \sum_{s=0}^{m-1} \sum_{\substack{i \neq j \\ i, j \in [sT, sT + T - 1]}} |i - j| \\
      &= 2 \sum_{\ell=1}^{mT-1} (mT - \ell)\ell - 2 \sum_{s=0}^{m-1} \sum_{\ell=1}^{T-1} (T - \ell)\ell \\
      &= \frac{1}{3}(mT - 1)mT(mT + 1) - \frac{1}{3}m(T - 1)T(T + 1),
\end{aligned}
\]
where the second-to-last equality follows from counting the number of pairs \((i, j)\) with \( j > i \) and \( j - i = \ell \).  

Substituting this expression into inequality~\eqref{weight2}, we obtain  
\[
r \geq 4t - \frac{4}{3} \cdot \sqrt{T(6t - T + 3)} + 2.
\]  
This completes the proof.

\end{proof}

\section{Upper bound on redundancy via explicit constructions}\label{sec:construction upperbound}
The first purpose of this part is to prove~\cref{thm:constructionGary}, which improves the previous upper bound $r_{\wt}(k,t)\le \frac{4t-2}{1-2\sqrt{\ln{(2t)}/(2t)}}$ for Hamming weight function. Then based on the framework we build, we can also derive the explicit constructions of FCCs for Hamming weight distribution.

Our construction establishes a framework building on the Gray codes. The improvement then comes from a careful selection of several linear codes with desired properties. Moreover, the construction is explicit. Here we first formally introduce some basics.

\subsection{Basics}
The term \emph{Gray code} was introduced in 1980 to describe a method for generating combinatorial objects such that successive objects differ in a prescribed way. Gray codes have been extensively studied, and in this work, we focus specifically on the classical binary reflected Gray code. For a broader discussion on Gray codes, see \cite{mutze2022combinatorial,savage1997survey}.

\begin{defn}[Gray code]
An $n$-bit Gray code for binary numbers is an ordering of all $n$-bit numbers $\left\{\aaa_i \right\}_{i=0}^{2^n-1}$ so that successive numbers (including the first and last) differ in exactly one position. 
\end{defn}

\begin{exam}
The following ordering is a $4$-bit Gray code. 
$$0000, 0001, 0011, 0010, 0110, 0111, 0101, 0100, 1100, 1101, 1111, 1110, 1010, 1011, 1001, 1000.$$
\end{exam}

Assume $L_n$ is an $n$-bit Gray code, we denote $x\cdot L_n$ as the ordering consisting of an additional prefix $x$ before every element belonging to $L_n$, and denote $L_n^{-1}$ as the reverse ordering of $L_n$. For example, if $L=11,22,33,44$, then $0\cdot L=011,022,033,044$ and $L^{-1}=44,33,22,11$. Then $\left\{0\cdot L_n, 1\cdot L_n^{-1} \right\}$ is an $(n+1)$-bit Gray code, which is called the binary reflected Gray code. 

A \emph{binary \([n, k, d]\) linear code} is a \(k\)-dimensional subspace of \(\mathbb{F}_2^n\) with minimum Hamming distance \(d\) between any two distinct codewords. The parameters \(n\), \(k\), and \(d\) represent the code length, dimension, and error-correcting capability, respectively. A linear code can be completely described by a \emph{generator matrix} \(G\), where \(G\) is a \(k \times n\) matrix whose rows span the code subspace.

\subsection{General framework via Gray codes}
We describe our construction as follows. Suppose there exists a binary \([n, k, d]\) linear code \(\mathcal{C}\) with a generator matrix \(\boldsymbol{G}\), which is a \(k \times n\) binary matrix. By appropriately reordering the columns of \(\boldsymbol{G}\), we assume that its first \(k\) columns form an identity matrix. By definition of the generator matrix, the code \(\mathcal{C}\) can be written as:
\[
\mathcal{C} = \{\boldsymbol{a}\boldsymbol{G} : \boldsymbol{a} \in \mathbb{Z}_2^k\}.
\]

Next, we partition each codeword in \(\mathcal{C}\) into two parts: the first \(k\) coordinates form a vector \(\boldsymbol{a} \in \mathbb{Z}_2^k\), while the remaining \(n - k\) coordinates form another vector \(\boldsymbol{b} \in \mathbb{Z}_2^{n - k}\). Specifically, for each \(\boldsymbol{a} \in \mathbb{Z}_2^k\), we express:
\[
\boldsymbol{a}\boldsymbol{G} = (\boldsymbol{a}, \boldsymbol{b}),
\]
where the first \(k\) columns of \(\boldsymbol{G}\) constitute the identity matrix \(I_k\).

We then order all vectors in \(\mathbb{Z}_2^k\) according to a \(k\)-bit Gray code, denoted by \(\{\boldsymbol{a}_i\}_{i=0}^{2^k - 1}\). This induces a corresponding ordering of the codewords in \(\mathcal{C}\), given by \(\{\boldsymbol{c}_i = (\boldsymbol{a}_i, \boldsymbol{b}_i)\}_{i=0}^{2^k - 1}\). We provide a simple example as follows to illustrate the above definitions.

\begin{exam}
    Consider a \(3\)-bit Gray code:
    \[
    000, 001, 011, 010, 110, 111, 101, 100.
    \]
    Let \(\mathcal{C}\) be a \([6, 3, 3]\) linear code with generator matrix:
    \[
    \boldsymbol{G} = \begin{pmatrix}
    1 & 0 & 0 & 0 & 1 & 1 \\
    0 & 1 & 0 & 1 & 0 & 1 \\
    0 & 0 & 1 & 1 & 1 & 0 \\
    \end{pmatrix}.
    \]
    By combining the above \(3\)-bit Gray code with the \([6, 3, 3]\) linear code \(\mathcal{C}\), the operation
    \[
    \{\boldsymbol{a}\boldsymbol{G}\}_{\boldsymbol{a} \in \mathbb{Z}_2^3} = \{(\boldsymbol{a}, \boldsymbol{b})\}_{(\boldsymbol{a}, \boldsymbol{b}) \in \mathcal{C}}
    \]
    results in the following ordering of the codewords in \(\mathcal{C}\):
    \begin{align*}
        &(0, 0, 0, 0, 0, 0), \\
        &(0, 0, 1, 1, 1, 0), \\
        &(0, 1, 1, 0, 1, 1), \\
        &(0, 1, 0, 1, 0, 1), \\
        &(1, 1, 0, 1, 1, 0), \\
        &(1, 1, 1, 0, 0, 0), \\
        &(1, 0, 1, 1, 0, 1), \\
        &(1, 0, 0, 0, 1, 1).
    \end{align*}
\end{exam}

We now formally describe our construction. Generally, let \(\mathcal{C}\) be a binary \([n, \lceil \log(2t+1) \rceil, 2t+1]\) linear code of size \(M := 2^{\lceil \log(2t+1) \rceil} \geq 2t+1\). The existence of such a linear code $\mathcal{C}$ will be discussed in~\cref{selection}. Using the \(\lceil \log(2t+1) \rceil\)-bit Gray code and the operation defined above, we can construct an ordering of the codewords \(\{\boldsymbol{c}_i = (\boldsymbol{a}_i, \boldsymbol{b}_i)\}_{i=0}^{M-1}\), such that the Hamming distance between \(\boldsymbol{a}_i\) and \(\boldsymbol{a}_{i+1}\) is exactly one for any \(0 \leq i \leq M-1\), where \(\boldsymbol{a}_M = \boldsymbol{a}_0\).

 For each \(\boldsymbol{u} \in \mathbb{Z}_2^k\), we define the encoding function as 
\[
\enc_{\textup{wt}}(\boldsymbol{u}) = (\boldsymbol{u}, \boldsymbol{p}_{\textup{wt}(\boldsymbol{u})}),
\]
where \(\boldsymbol{p}_i\) is given by:
\[
\boldsymbol{p}_i = 
\begin{cases} 
\boldsymbol{b}_i & \text{for } 0 \leq i \leq M-1, \\
\boldsymbol{b}_{i \bmod M} & \text{for } i \geq M.
\end{cases}
\]

Next, we need to show that our construction satisfies the properties of FCCs listed in~\cref{lem:weight functionalcoding}.
\begin{claim}
For any \(i \neq j\in [0,k] \), 
\[
d(\boldsymbol{p}_i, \boldsymbol{p}_j) + d(\uu_{i},\uu_{j}) \geq 2t+1.
\]
\end{claim}
\begin{poc}
    Note that we can assume \(d(\uu_{i},\uu_{j})=|i - j| \leq 2t\), since the inequality holds trivially when \(|i - j| \geq 2t + 1\). Due to the symmetry of \(i\) and \(j\), we can further assume without loss of generality that \(1 \leq i - j \leq 2t\). Then, by definition, we have:
    \begin{align*}
        d(\pp_i,\pp_j)&= d(\bb_{i\bmod M},\bb_{j\bmod M}) \\
        &= d(\cc_{i\bmod M},\cc_{j\bmod M}) - d(\aaa_{i\bmod M},\aaa_{j\bmod M}) \\
        &\geq 2t+1 - \sum _{s=j}^{i-1}d(\aaa_{s\bmod M},\aaa_{s+1\bmod M}) \\
        &= 2t+1-|i-j|.
    \end{align*}
    Here, the inequality follows from the fact that $\cc_{i\bmod M}\neq \cc_{j\bmod M}$ and from the triangle inequality under Hamming metric. 
\end{poc}

\subsection{Selection of desired linear codes: Proof of~\cref{thm:constructionGary}}\label{selection}
Based on the above framework, for a given \(k = \lceil \log(2t+1) \rceil\) and \(d = 2t+1\), achieving a better bound on the optimal redundancy requires selecting a binary \([n, k, d]\) linear code with \(n\) as small as possible.  
A fundamental limitation on binary \([n, k, d]\) linear codes is given by the following bound, discovered in~\cite{griesmer1960bound}.
\begin{lemma}[\cite{griesmer1960bound}]\label{lemma:Griesmer}
    For any binary \([n, k, d]\) linear code, we have
    \[
        n \geq \sum_{i=0}^{k-1} \left\lceil \frac{d}{2^i} \right\rceil.
    \]
\end{lemma}
Here, we present several binary linear codes that attain the bound in~\cref{lemma:Griesmer}, and provide better upper bounds on the optimal redundancy under various parameter conditions for the Hamming weight function. The chosen linear codes here include simplex codes and Belov-type codes~\cite{macwilliams1977theorypart1}.

\begin{enumerate}
    \item[\textup{(1)}] \textbf{Simplex codes:}
For each positive integer \(m\), let \(t = 2^{m-1} - 1\). It is known that there exists a simplex code, which is a binary \([2^m - 1, m, 2^{m-1}]\) linear code, with a generator matrix \(\boldsymbol{G}_m\). The \(2^m - 1\) columns of \(\boldsymbol{G}_m\) consist of all nonzero vectors in \(\mathbb{Z}_2^m\). For example,  
\[
\boldsymbol{G}_3 = 
\begin{pmatrix}
    0 & 0 & 0 & 1 & 1 & 1 & 1 \\
    0 & 1 & 1 & 0 & 0 & 1 & 1 \\
    1 & 0 & 1 & 0 & 1 & 0 & 1
\end{pmatrix}.
\]  
The simplex code is, in fact, the dual of the well-known Hamming code. Furthermore, we can construct a new code by modifying \(\boldsymbol{G}_m\). Specifically, repeat \(\boldsymbol{G}_m\) once to form \((\boldsymbol{G}_m, \boldsymbol{G}_m)\), and then remove the first column to obtain a new matrix \(\boldsymbol{G}_m'\). Let \(\mathcal{C}'\) be the linear code generated by \(\boldsymbol{G}_m'\). It is straightforward to verify that \(\mathcal{C}'\) is a binary \([2^{m+1} - 3, m, 2^m - 1]\) linear code, which achieves the Griesmer bound.  

By simple calculations, the redundancy of the weight FCC is given by  
\[
r = 2^{m+1} - 3 - m = 4t - \lceil \log t \rceil.
\]
    
    \item[\textup{(2)}]\textbf{Belov-type codes:} 
Note that the above construction requires \(t+1 = 2^{m-1}\). Here, we present a more general construction that provides a good upper bound on the optimal redundancy (slightly worse than \(4t - \lceil \log t \rceil\)) while allowing greater flexibility for the parameter \(t\). Let \(m = \lceil \log(2t+1) \rceil\). There exist a unique integer \(1 \leq p \leq m-1\) and unique integers \(m > u_1 > \cdots > u_p = 1\) such that  
\[
\sum_{i=1}^{p} 2^{u_i} = 2^{m+1} - 4t - 2.
\]  

When \(\sum_{i=1}^{\min\{3, p\}} u_i \leq 2m\), there exists a \([4t + p, m, 2t + 1]\) linear code, referred to as a Belov-type code, which also achieves the Griesmer bound. Especially, when $u_1\le \frac{2}{3}m+1$, condition \(\sum_{i=1}^{\min\{3, p\}} u_i \leq 2m\) must be satisfied. That is to say, when $2^{\ceil{\log(2t+1)}}-(2t+1)\le 2^{\frac{2}{3}\ceil{\log(2t+1)}}$, \([4t + p, m, 2t + 1]\) linear code must exist. By simple calculations, the redundancy of the weight FCC is given by  
\[
4t + p - \lceil \log(2t + 1) \rceil,
\]  
where \(p\) can be interpreted as the number of ones in the binary representation of \(2^{\lceil \log(2t + 1) \rceil} - (2t+1)\). 
\end{enumerate}

We believe that there are more linear codes that can be selected as candidates here, which might be helpful to further obtain better upper bound on $r_{\wt}(k,t)$. However, since we have already presented our construction intent in a relatively complete manner, we will not pile up more examples. For the selection of more nice binary linear codes, we suggest referring to the content in these references~\cite{brouwer1998bounds,hill1999survey,Huffman2003correcting}.

\subsection{A general framework for Hamming weight distribution: Proof of~\cref{thm:constructionweightdistribution}}

The main purpose of this part is to improve the upper bound for Hamming weight distribution function via explicit constructions.

Let \(f(\uu)=\Delta_T(\uu) = \left\lfloor \frac{\textup{wt}(\uu)}{T} \right\rfloor\) for any \(T \leq 2t\). To simplify notation, let \(E := \left\lfloor \frac{k}{T} \right\rfloor + 1 = |\textup{Im}(\Delta_T(\mathbb{Z}_2^k))|\). The framework can be described as follows:

\begin{theorem}\label{thm:constructionweightdistribution2}
Let \(k\) and \(t\) be positive integers. Let \(\pp_i = (0^{T-i-1}1^i) \in \mathbb{Z}_2^{T-1}\) for \(0 \leq i \leq T-1\), and \(\pp_i = \pp_{i \bmod T}\) for \(T\le i\le E-1\). Suppose there exists an integer \(s\) and a sequence of vectors \(\{\qq_i\}_{i=0}^{E-1} \subseteq \mathbb{Z}_2^s\) such that  
\[
d(\qq_i, \qq_j) \geq 2t + 1 - T \cdot |i - j| \quad \text{for } 0 \leq i \neq j \leq E-1.
\]  
Then, for any \(\uu \in \mathbb{Z}_2^k\), the encoding function 
\[
\enc_{\Delta_T}(\uu) = (\uu, \qq_{\Delta_T(\uu)}, \pp_{\textup{wt}(\uu)})
\]  
is a valid encoding function of function-correcting code for the weight distribution function with redundancy $s+T-1$.
\end{theorem}

\begin{proof}[Proof of~\cref{thm:constructionweightdistribution2}]
Let \(\uu, \vv \in \mathbb{Z}_2^k\) with \(\Delta_T(\uu) \neq \Delta_T(\vv)\). By the definition of FCCs, it suffices to prove that \(d(\enc(\uu), \enc(\vv)) \geq 2t + 1\). First, suppose \(|\wt(\vv) - \wt(\uu)| \geq 2t + 1\). In this case, we are immediately done since \(d(\enc(\uu), \enc(\vv)) \geq d(\uu, \vv) \geq |\wt(\vv) - \wt(\uu)| \geq 2t + 1\).  

Now, assume \(|\wt(\vv) - \wt(\uu)| < 2t + 1\). Without loss of generality, let  
\[
\wt(\uu) = mT + i \quad \text{and} \quad \wt(\vv) = (m + n)T + j,
\]  
for some integers \(m \geq 0\), \(n \ge 1\), and \(i, j \in [0, T-1]\).  

\begin{enumerate}
\item If \(i \geq j\), we then have  
\[
\begin{aligned}
d(\enc(\uu), \enc(\vv)) &= d(\uu, \vv) + d(\qq_m, \qq_{m+n}) + d(\pp_i, \pp_j) \\
&\geq \wt(\vv) - \wt(\uu) + \big(2t + 1 - T \cdot |m - (m + n)|\big) + (i - j) \\
&= nT + j - i + 2t + 1 - nT + i - j \\
&= 2t + 1,
\end{aligned}
\]  
where the inequality follows from \(d(\uu, \vv) \geq |\wt(\vv) - \wt(\uu)|\) and the distance properties of \(\qq_m, \qq_{m+n}\).  

\item If $i<j$, then we have  
\[
\begin{aligned}
d(\enc(\uu), \enc(\vv)) &= d(\uu, \vv) + d(\qq_m, \qq_{m+n}) + d(\pp_i, \pp_j) \\
&\geq \wt(\vv) - \wt(\uu) + \big(2t + 1 - T \cdot |m - (m + n)|\big) + (j - i) \\
&= nT + j - i + 2t + 1 - nT + j - i \\
&= 2t + 1 + 2(j - i) \\
&> 2t + 1,
\end{aligned}
\]  
where the first inequality follows from \(d(\uu, \vv) \geq |\wt(\vv) - \wt(\uu)|\).  
\end{enumerate}

Thus, in both cases, we have \(d(\enc(\uu), \enc(\vv)) \geq 2t + 1\). Furthermore, the redundancy of our encoding function is equal to the sum of the lengths of \(\{\qq_i\}_{i=0}^{E-1}\) and \(\{\pp_i\}_{i=0}^{T-1}\). This completes the proof.

\end{proof}
We then show the FCCs by selecting certain sequences of vectors $\{\qq_{i}\}_{i=0}^{E-1}$.
\begin{enumerate}
    \item[\textup{(1)}] 

When \(T \geq t + 1\), we can define the redundancy vectors \(\qq_i\) as follows:
\[
\qq_i = 
\begin{cases} 
0^{2t-T+1}, & \text{if \(i\) is odd}; \\ 
1^{2t-T+1}, & \text{if \(i\) is even}.
\end{cases}
\]
By~\cref{thm:constructionweightdistribution}, the redundancy of our encoding function is  
\[
(2t - T + 1) + (T - 1) = 2t.
\]
Furthermore, this explicit construction is optimal. To see this, observe that \(|\textup{Im}(\Delta_T)| \geq 2\). By~\cref{cor1}, the redundancy satisfies \(r_{\Delta_T}(k, t) \geq 2t\), confirming the optimality of our construction.

    \item[\textup{(2)}] 

For the general case \(T \leq t\), let \(z := \ceil{\frac{2t+1}{T}}\). Suppose that \(2^{\lceil \log z \rceil} - z \leq 2^{\frac{2}{3}\lceil \log z \rceil}\). Then, by~\cref{thm:constructionGary}, there exists a family of vectors \(\{\qq'_i\}_{i=0}^{E-1}\), referred to as redundancy vectors, of length 
\[
\frac{4t + 2}{T} - 2 + p - \lceil \log z \rceil,
\]
where \(p\) denotes the number of ones in the binary representation of \(2^{\lceil \log z \rceil} - z\). These vectors satisfy the property \(d(\qq'_i, \qq'_j) \geq z - |i - j|\) for all \(i, j\).
    
Based on the selection of \(\{\qq_{i}'\}_{i=0}^{E-1}\), we construct a new family \(\{\qq_{i}\}_{i=0}^{E-1}\) by setting \(\qq_{i}\) to be the \(T\)-fold repetition of \(\qq_{i}'\). Specifically, \(\qq_{i} = (\qq_{i}', \qq_{i}', \ldots, \qq_{i}')\) with \(T\) repetitions. This results in a new construction of redundancy vectors of length 
\[
\left(\frac{4t + 2}{T} - 2 + p - \lceil \log z \rceil\right)\cdot T+T-1=4t + (p - 1)T - \lceil \log z \rceil\cdot T + 1,
\]
 where \(p\) represents the number of ones in the binary representation of \(2^{\lceil \log z \rceil} - z\).
    
Specifically, when \(\frac{2t+1}{T} = 2^m - 1\) for some integer \(m\), the redundancy of our construction simplifies to \(4t - mT + 1 = 4t - \lceil \log \frac{2t+1}{T} \rceil \cdot T + 1\). Note that \(2^{\lceil \log z \rceil} - z\le 2^{\frac{2}{3}\lceil \log z \rceil} \leq 2^{\lceil \log z \rceil - 1}\), which implies \(p \leq \lceil \log z \rceil\). Consequently, when $t$ is sufficiently large and $T=o(t)$, the redundancy of this construction becomes  
\[
4t + (p - 1)T - \lceil \log z \rceil\cdot T + 1 \le 4t-T = 4t - o(t),
\]  
which is asymptotically optimal according to Theorem~\ref{thm:NewDistribution}.

\end{enumerate}

\section{Concluding remarks}\label{sec:Concluding}
In this paper, we study the optimal redundancy of function-correcting codes (FCCs) for two important function classes: the Hamming weight function and the Hamming weight distribution function. We establish near-optimal lower bounds and construct explicit codes that improve upon previously known upper bounds.  

Our lower bounds are derived using a combination of double counting arguments and structural analysis. This approach also extends to function-correcting codes in symbol-pair channels~\cite{xia2024function}, which arise in high-density data storage where read operations return pairs of consecutive symbols rather than individual ones~\cite{Cassuto2011symbolpair,Ding2018MDSsymbol,Elishco2020boundsymbolpair,li2017MDSsymbol}. Unlike the Hamming distance, which accounts for single-symbol errors, the symbol-pair distance captures errors in consecutive pairs, providing a more robust framework for error detection and correction. Formally, for two codewords \( \boldsymbol{x} = (x_1, x_2, \dots, x_n) \) and \( \boldsymbol{y} = (y_1, y_2, \dots, y_n) \) of length \( n \), their symbol-pair distance \( d_p(\boldsymbol{x}, \boldsymbol{y}) \) is defined as the number of positions \( i \in \{1, 2, \dots, n\} \) where the consecutive symbol pairs \( (x_i, x_{i+1}) \) and \( (y_i, y_{i+1}) \) differ, with indices taken modulo \( n \) to account for cyclicity. That is,
\[
d_p(\boldsymbol{x}, \boldsymbol{y}) = \left| \{ i \mid (x_i, x_{i+1}) \neq (y_i, y_{i+1}) \} \right|.
\]

Motivated by these challenges, Xia, Liu, and Chen~\cite{xia2024function} introduced function-correcting symbol-pair codes (FCSPCs) to minimize redundancy while improving information storage efficiency. Similar to FCCs, FCSPCs ensure that key attributes of the transmitted message can be accurately recovered despite errors. A central open problem in this setting is determining the optimal redundancy required for such codes.  

\begin{defn}
Let \( k \), \( r \), and \( t \) be positive integers. An encoding function \( \enc: \Z_2^k \to \Z_2^{k+r} \), defined by  
\[
\enc(\uu) = (\uu, \pp(\uu))
\]  
for \( \uu \in \Z_2^k \), is said to define a function-correcting symbol-pair code (FCSPC) for a function \( f: \Z_2^k \to \im(f) \) if, for all \( \uu_1, \uu_2 \in \Z_2^k \) with \( f(\uu_1) \neq f(\uu_2) \), the pairwise distance satisfies  
\[
d_p(\enc(\uu_1), \enc(\uu_2)) \geq 2t + 1.
\]  
The optimal redundancy \( r_p^f(k, t) \) is the smallest integer \( r \) for which there exists an encoding function \( \enc: \Z_2^k \to \Z_2^{k+r} \) satisfying this condition for \( f \).  
\end{defn}  

In~\cite{xia2024function}, the following upper and lower bounds on the optimal redundancy were established:  

\begin{lemma}[\cite{xia2024function}]
    For any integers \(k\) and \(t \geq 6\) with \(t < k \leq (2t-1)^2\),  
    \[
    \frac{20t^3 - 20t}{9(t+1)^2} \leq r_{p}^{\wt_p}(k,t) \leq \frac{4t-4}{1 - 2\sqrt{\frac{\ln(2t-1)}{2t-1}}}.
    \]
\end{lemma}  

We can improve the lower bound from \((\frac{20}{9}-o(1))t\) to \((\frac{8}{3}-o(1))t\). The proof follows similar arguments to those used in our lower bound analysis for FCCs, thus we omit the full details here.  

\begin{theorem}\label{thm:lowerboundpairweight}
    For any integers \(k > t\ge 6\),  
    \begin{equation*}
        r_{p}^{\wt_p}(k,t) \geq \frac{8t}{3} - \frac{8}{9}\sqrt{6t-4} - \frac{4}{3}.
    \end{equation*}
\end{theorem} 

Additionally, developing more efficient FCCs for various natural and practical functions remains an important direction for future work, alongside establishing tight theoretical bounds on redundancy. An alternative approach to bounding redundancy is through the graph-theoretic framework. 
An \emph{independent set} in a graph \( G \) is a set of vertices such that no two vertices in the set are adjacent. The \emph{independence number} \( \alpha(G) \) is the size of the largest independent set in \( G \). It was also noted in~\cite{lenz2023function} that estimating \( r_f(k,t) \) is equivalent to determining the independence number of a graph defined in terms of the function \( f \).
\begin{defn}\label{def:GraphTheoretical}
	Let \( G_f(k,t,r) \) be the graph whose vertex set is \( V = \{0,1\}^k \times \{0,1\}^r \), where each vertex is of the form \( \xx = (\uu, \pp) \in \{0,1\}^{k+r} \). Two vertices \( \xx_1 = (\uu_1, \pp_1) \) and \( \xx_2 = (\uu_2, \pp_2) \) are adjacent if either \( \uu_1 = \uu_2 \), or both \( f(\uu_1) \neq f(\uu_2) \) and \( d(\xx_1, \xx_2) \leq 2t \) hold. Define \( \gamma_f(k,t) \) as the smallest integer \( r \) such that there exists an independent set of size \( 2^k \) in \( G_f(k,t,r) \).
\end{defn}
It was shown in~\cite{lenz2023function} that \( \gamma_f(k,t) = r_f(k,t).\) However, for the Hamming weight function and Hamming weight distribution, the corresponding graphs are too dense to yield effective independence number estimates, even with advanced tools from extremal graph theory~\cite{1980Ramsey,campos2023newlowerboundsphere}. Exploring the applicability of this approach to other functions is a promising avenue for future research.

\bibliographystyle{abbrv}
\bibliography{ImprovedGV}
\end{document}